\documentclass[sigconf]{acmart}

\usepackage{algorithm}
\usepackage{algpseudocode}
\usepackage{subcaption}

\copyrightyear{2026}
\acmYear{2026}
\setcopyright{cc}
\setcctype{by}
\acmConference[ICPP '26]{Proceedings of the 55th International Conference on Parallel Processing}{September 28-October 01, 2026}{Singapore, Singapore}
\acmBooktitle{Proceedings of the 55th International Conference on Parallel Processing (ICPP '26), September 28-October 01, 2026, Singapore, Singapore}
\acmDOI{10.1145/3832810.3832873}
\acmISBN{979-8-4007-2657-6/2026/09}

\graphicspath{{figures/}}

\begin{document}

\title{Scalable Exact Path Selection via Structure-Aware Search for Virtual Payment Channels}

\author{Jiangnan Luo}
\orcid{0009-0003-7111-4486}
\affiliation{%
  \department{State Key Laboratory of Novel Software Technology}
  \institution{Nanjing University}
  \city{Nanjing}
  \country{China}
}
\email{jiangnanluo@smail.nju.edu.cn}

\author{Zhebei Shen}
\orcid{0009-0000-7043-3731}
\affiliation{
  \institution{Zhejiang University}
  \city{Hangzhou}
  \country{China}
}
\email{shenzhebei@zju.edu.cn}

\author{Yuan Zhang}
\orcid{0000-0001-9682-5231}
\affiliation{%
  \department{State Key Laboratory of Novel Software Technology}
  \institution{Nanjing University}
  \city{Nanjing}
  \country{China}
}
\email{zhangyuan@nju.edu.cn}

\author{Sheng Zhong}
\orcid{0000-0002-6581-8730}
\affiliation{%
  \department{State Key Laboratory of Novel Software Technology}
  \institution{Nanjing University}
  \city{Nanjing}
  \country{China}
}
\email{zhongsheng@nju.edu.cn}

\begin{abstract}
Virtual Payment Channels (VPCs) enable efficient off-chain transactions in Payment Channel Networks (PCNs), but their performance depends on selecting high-quality underlying paths. Existing approaches either rely on simplified metrics or incur high computational cost.

We study VPC path selection under generalized monotone metrics and propose a structure-aware exact solver based on quadtree search. By exploiting monotonicity and distance plateau properties, our method prunes large regions of the capacity-constrained search space while preserving optimality, significantly reducing the number of shortest-path computations.

We further instantiate the framework with a composite metric that integrates economic cost and security risk, enabling flexible trade-offs across application scenarios. Experiments on synthetic graphs and real-world Lightning Network topologies (up to 12,552 nodes) show 2--5 orders of magnitude speedup over prior work, with consistent sub-100ms latency.
\end{abstract}

\ccsdesc[500]{Theory of computation~Shortest paths}
\ccsdesc[500]{Networks~Network algorithms}
\ccsdesc[300]{Security and privacy~Distributed systems security}
\ccsdesc[100]{Applied computing~Digital cash}

\keywords{Payment channel networks, virtual payment channels, path selection, Lightning Network, shortest-path algorithms, quadtree search}

\maketitle

\section{Introduction}
\label{sec:intro}

The rapid adoption of decentralized cryptocurrencies, most notably Bitcoin, has exposed fundamental scalability limitations, such as low throughput and high confirmation latency. To address these challenges, Layer-2 solutions \cite{rebello2024survey,rao2024scalability,dotan2021survey}, particularly Payment Channel Networks (PCNs) such as the Lightning Network (LN) \cite{poon2016bitcoin}, have been proposed and widely deployed. By offloading transactions from the blockchain to a network of bidirectional payment channels, PCNs enable near-instantaneous payments with significantly reduced fees. However, their performance critically depends on efficient routing and channel utilization in a dynamic and resource-constrained network.

A promising extension is Virtual Payment Channels (VPCs) \cite{dziembowski2019perun, aumayr2021bitcoin}, which allow two nodes to establish a logical direct channel over a sequence of existing channels, referred to as the \textit{underlying path}. By bypassing repeated multi-hop routing, VPCs can reduce cumulative fees, improve reliability, and mitigate certain security risks. The effectiveness of a VPC, however, is fundamentally determined by the quality of its underlying path.

Selecting an optimal underlying path is challenging: it requires jointly balancing heterogeneous and often competing objectives, such as maximizing transferable capacity while minimizing economic cost and exposure to adversarial risks. Existing approaches typically focus on specific aspects, such as topological efficiency or security \cite{malavolta2017concurrency}, but often rely on simplified models or incur significant computational overhead. This leaves a gap for a unified and efficient solution that flexibly captures diverse goals while remaining practical at scale. Notably, unlike per-payment routing, VPC establishment is an infrequent decision that locks collateral for the channel's entire lifetime, and a suboptimal underlying path penalizes every subsequent payment; this justifies exact optimization---provided it is fast enough for interactive use, which our solver achieves.

In this paper, we propose a generalized metric framework for VPC path selection that addresses these challenges by decoupling optimization logic from metric design.
Our main contributions are:

\begin{itemize}
\item \textbf{Structure-aware exact search.}
We study VPC path selection in the capacity-constrained search space and design an exact quadtree-based algorithm that exploits monotonicity and distance plateau structures to prune large regions while preserving optimality.
\item \textbf{Effective pruning mechanism.}
We develop synergistic pruning strategies that significantly reduce the number of shortest-path evaluations, enabling practical efficiency even at large scales.
\item \textbf{Scalable performance.}
Experiments on real and synthetic networks show 2--5 orders of magnitude speedup over prior work with consistent sub-100ms latency.
\item \textbf{Metric instantiation.}
We instantiate the framework with a composite metric (SECER) that integrates economic cost and security risk, demonstrating its flexibility in supporting different optimization objectives.
\end{itemize}

The remainder of this paper is organized as follows. Section \ref{sec:background} provides the necessary background on VPCs and reviews related work. Section \ref{sec:framework} formalizes the generalized VPC path selection problem. Section \ref{sec:solver} details our high-performance quadtree-based solver and its pruning strategies. Section \ref{sec:metric} introduces the SECER metric for balancing economy and security. Section \ref{sec:evaluation} presents the experimental results, and Section \ref{sec:conclusion} concludes the paper.

\section{Background}
\label{sec:background}

\subsection{Virtual Payment Channels}
\label{subsec:vpc_def}


Virtual Payment Channels (VPCs) enable non-adjacent nodes in a Payment Channel Network (PCN) to establish a direct off-chain relationship by leveraging a sequence of existing channels, referred to as the \textit{underlying path} \cite{dziembowski2019perun, aumayr2021bitcoin}. The endpoints can then transact as if connected by a direct channel, avoiding repeated multi-hop routing for each payment.

By locking collateral along the underlying path, a VPC provides dedicated capacity between the endpoints and, once established, supports instant and private payments without the continuous involvement of intermediate nodes. Its effectiveness, however, critically depends on the choice of the underlying path: candidate paths vary significantly in capacity, cost, and reliability, making path selection a key challenge in VPC construction.

\subsection{Related Work}

Research on VPCs has primarily focused on cryptographic protocols and security proofs, often assuming the underlying path is either pre-selected or discovered via standard routing heuristics \cite{li2023efficient, kiayias2021elmo}. Only recently has the problem of optimizing the underlying path itself gained some attention.

\paragraph{Capacity-Oriented Optimization.}
Thor \cite{wei2024thor} was the first to rigorously formalize the trade-off between VPC capacity and intermediary overhead by introducing the Capacity to Number of Intermediaries Ratio (CNIR). For an underlying path $P$, CNIR is defined as the ratio between the total capacity of the VPC and the number of intermediate nodes. Thor proposes an exact algorithm for the maximum CNIR problem by iterating over all possible capacity thresholds and performing breadth-first searches.

Despite its optimality, Thor suffers from high computational complexity, namely $O(m^2(n+m))$, where $n$ is the number of nodes and $m$ is the number of channels. This makes it difficult to scale to large PCNs such as the modern Lightning Network, which consists of tens of thousands of nodes and channels \cite{seres2020topological,valko2025topology}. Moreover, CNIR relies on hop count as a proxy for cost, which overlooks the heterogeneity of intermediaries in terms of fees and reliability.

\paragraph{Economic and Security-Aware Optimization.}
Aumayr et al. \cite{aumayr2025optimizing} study VPC construction under adversarial conditions by incorporating node risk into the optimization process. Their work formulates a global optimization problem that balances VPC construction costs with long-term transaction costs. To compute optimal strategies, they employ an Integer Linear Program (ILP), which is computationally expensive and mainly suitable for offline planning.

For individual node pairs, they further propose a greedy heuristic to approximate the trade-off between cost and security. While their approach introduces adversarial considerations into VPC path selection, it primarily treats node risk as a binary exclusion criterion (skipping nodes assumed corrupted under an attacker's budget) rather than a continuous metric integrated into path selection.

In contrast to prior work, our approach aims to provide a unified and efficient framework that can incorporate multiple optimization objectives, including capacity, cost, and security, while remaining scalable to modern PCNs.

\section{Generalized Metric Model for Path Selection}
\label{sec:framework}

Building on these concepts, we propose a generalized metric model for VPC path selection. This framework transitions from simple hop-count optimization to a multi-dimensional evaluation of underlying paths, allowing for a flexible balance between capacity, economy, and security.

\subsection{Problem Formulation}
The core objective of our framework is to identify an \textit{underlying path} (ULP) $P$ that maximizes a generalized efficiency ratio. Formally, given a source node $s$ and a target node $t$ in a payment channel network $G = (V, E)$, we seek to solve:
\begin{equation}
    \max_{P \in \mathcal{P}_{s \to t}} \Phi(P) = \frac{\delta(P^{s\to t}) + \delta(P^{t\to s})}{D(P)}
\end{equation}
where:
\begin{itemize}
    \item $\mathcal{P}_{s \to t}$ is the set of all simple paths between $s$ and $t$.
    \item $\delta(P^{s\to t})$ and $\delta(P^{t\to s})$ are the \textit{maximum liquefiable balances} (MLB) of the path $P$ in the forward and backward directions, respectively.
    \item $D(P)$ is a generalized distance function that evaluates the "cost" or "penalty" associated with path $P$.
\end{itemize}

\subsection{Properties of the Distance Function}

To ensure that the optimization problem is computationally solvable, we define $D(P)$ as a generalized cost function that evaluates the "distance" of a path $P$. While the simplest form of $D(P)$ is the summation of static edge weights, our framework is designed to accommodate more complex, realistic cost structures such as compounding routing fees and cumulative security risks.

The essential requirement for $D(P)$ in our framework is that it must exhibit \textit{Optimal Substructure} and \textit{Monotonicity}. Specifically, let $P_{u \to v}$ be a path from $u$ to $v$. For any intermediate node $w \in P_{u \to v}$, the cost of the prefix path must be a monotonic function of its sub-paths. This ensures that the "shortest" path (the path that minimizes $D(P)$ under certain capacity constraints) can be efficiently computed using generalized shortest-path algorithms, such as Dijkstra's algorithm \cite{dijkstra2022note} or its variants for payment networks \cite{roos2018settling}.

For example, in the context of economic costs, $D(P)$ can represent the total fee required to send a unit of balance. Even if fees are compounded (i.e., fees are charged on top of fees for downstream hops), the shortest-path property holds as long as the fee functions are non-negative and monotonic with respect to the transaction amount. By requiring $D(P)$ to be a "shortest-path distance," we maintain a high degree of modeling flexibility without sacrificing algorithmic efficiency.

\subsection{Generalization of the CNIR Metric}
The proposed framework is a strict generalization of the Capacity to Number of Intermediaries Ratio (CNIR) proposed in Thor \cite{wei2024thor}, which measures the overhead of a path $P = (v_0, v_1, \dots, v_k)$ by its number of intermediate nodes $d = k-1$. Setting $D_{\text{CNIR}}(P) = d$---i.e., assigning each intermediate node a unit weight---reduces the objective function $\Phi(P)$ exactly to the original CNIR metric. Generalizing $D(P)$ from an intermediary count to a shortest-path distance accommodating compounding fees and heterogeneous risks lets our framework model real-world Lightning Network dynamics while retaining the rigorous optimization properties of CNIR.

\section{Efficient Quadtree-based Exact Solver}
\label{sec:solver}

To solve the generalized optimization model proposed in Section \ref{sec:framework}, relying on the nested-loop enumeration approach used by Thor is computationally prohibitive, especially when the cost function $D(P)$ incorporates complex calculations.
Therefore, we propose an efficient quadtree-based exact solver to solve this problem.

\subsection{Problem Transformation}
The objective function $\Phi(P) = \frac{\delta(P^{s\to t}) + \delta(P^{t\to s})}{D(P)}$ involves both capacities and distances, which are coupled. However, we can decouple them by observing that for any fixed forward capacity threshold $a$ and backward capacity threshold $b$, the ``best'' possible path is the one that minimizes the distance $D(P)$ while satisfying the capacity constraints.

Specifically, for any pair $(a, b)$, let $G_{a,b} = (V, E_{a,b})$ be the restricted subgraph where $E_{a,b} = \{e \in E \mid w_{s\to t}(e) \ge a \wedge w_{t\to s}(e) \ge b\}$. We define the minimum distance function $D(a, b)$ as:
\begin{equation}
    D(a, b) = \min_{P \subseteq G_{a,b}} D(P)
\end{equation}

\begin{lemma}[Monotonicity of $D(a, b)$]
    \label{lemma:mono}
    The minimum distance function $D(a, b)$ is monotonically non-decreasing with respect to $a$ and $b$. That is, for any $a \le a'$ and $b \le b'$, $D(a, b) \le D(a', b')$.
\end{lemma}
\begin{proof}
    For $a \le a'$ and $b \le b'$, the set of paths in $G_{a', b'}$ is a subset of those in $G_{a, b}$. Since $D(a, b)$ and $D(a', b')$ are the minimum distances over these sets respectively, taking the minimum over a larger set yields a smaller or equal value, hence $D(a, b) \le D(a', b')$.
\end{proof}

The global optimization problem can then be reformulated as a search over the two-dimensional space of discrete edge weights:
\begin{equation}
    \label{equation:transform}
    \max_{P \in \mathcal{P}} \Phi(P) = \max_{a, b \in W} \frac{a + b}{D(a, b)}
\end{equation}

To ensure correctness, our solver requires three conditions on the metric: (i) \emph{monotonicity}---$D(P)$ is path-monotone and can be exactly optimized by a shortest-path oracle (e.g., Dijkstra's algorithm); (ii) \emph{stability}---$D(P)$ is positive for all feasible paths and independent of the searched capacity thresholds $(a, b)$; and (iii) \emph{discrete bottlenecks}---the bottleneck capacities of any optimal path belong to the set of unique edge capacities $W$.

\begin{theorem}[Equivalence of Optimization]
    \label{thm:equivalence}
    Under the aforementioned conditions, the optimal value of the path-based objective function $\Phi(P)$ is identical to the optimal value of the capacity-space search in equation \ref{equation:transform}.
\end{theorem}
\begin{proof}
    For any path $P$, let $a = \delta(P^{s \to t})$ and $b = \delta(P^{t \to s})$. Since $P \subseteq G_{a,b}$, $D(a,b) \le D(P)$, which implies $\Phi(P) = \frac{a+b}{D(P)} \le \frac{a+b}{D(a,b)} \le \max_{a,b} \frac{a+b}{D(a,b)}$.
    Conversely, for any pair $(a,b)$, let $P$ be the shortest path in $G_{a,b}$. Then $\delta(P^{s \to t}) \ge a$, $\delta(P^{t \to s}) \ge b$, and $D(P) = D(a,b)$, so $\Phi(P) = \frac{\delta(P^{s \to t}) + \delta(P^{t \to s})}{D(P)} \ge \frac{a+b}{D(a,b)}$. Thus, both maxima must coincide.
\end{proof}

This transformation shifts the focus from enumerating an exponential number of paths to searching within the finite space of edge weights.

\subsection{Universal Quadtree Exact Search Algorithm}
To efficiently navigate the capacity space and identify the optimal underlying path, we propose a universal quadtree-based exact search algorithm, as detailed in Algorithm \ref{alg:quadtree}. Unlike exhaustive search methods that iterate through every possible weight pair, our algorithm leverages the mathematical properties established below to skip vast regions of the search space.

\begin{algorithm}[t]
\caption{Universal Quadtree-based Exact Solver}
\label{alg:quadtree}
\begin{algorithmic}[1]
\Procedure{SolveVPC}{$s, t, G$}
    \State $W \gets$ sorted unique capacities of $\{w_{s\to t}(e) \mid e \in E\} \cup \{w_{t\to s}(e) \mid e \in E\}$
    \State $\Phi^* \gets 0, P^* \gets \text{null}$
    \State \Call{QuadSearch}{1, $|W|$, 1, $|W|$}
    \State \Return $\Phi^*, P^*$
\EndProcedure

\Procedure{QuadSearch}{$i_1, i_2, j_1, j_2$}
    \State \textbf{if} $i_1 > i_2 \lor j_1 > j_2$ \textbf{then} \Return \Comment{Degenerate block}
    \State $a_{min}, a_{max} \gets W[i_1], W[i_2]$
    \State $b_{min}, b_{max} \gets W[j_1], W[j_2]$
    \If{$\Phi^* = 0$} $d_{lim} \gets \infty$
    \Else \ $d_{lim} \gets (a_{max} + b_{max}) / \Phi^*$
    \EndIf \Comment{Dynamic pruning threshold}
    \State $(d_{min}, P_{min}) \gets \text{ShortestPath}(G_{a_{min}, b_{min}}, s, t, d_{lim})$
    \If{$d_{min} = \infty$}
        \State \Return
    \EndIf
    \State $(d_{max}, P_{max}) \gets \text{ShortestPath}(G_{a_{max}, b_{max}}, s, t, d_{lim})$
    \If{$d_{min} = d_{max}$} \Comment{Plateau found}
        \State $\Phi_{curr} \gets (a_{max} + b_{max}) / d_{max}$
        \If{$\Phi_{curr} > \Phi^*$}
            \State $\Phi^* \gets \Phi_{curr}, P^* \gets P_{max}$
        \EndIf
        \State \Return \Comment{Sparsity-based pruning (Thm \ref{thm:sparsity})}
    \EndIf
    \State $i_{mid} \gets \lfloor (i_1 + i_2) / 2 \rfloor, j_{mid} \gets \lfloor (j_1 + j_2) / 2 \rfloor$
    \State \Call{QuadSearch}{$i_1, i_{mid}, j_1, j_{mid}$}
    \State \Call{QuadSearch}{$i_{mid}+1, i_2, j_1, j_{mid}$}
    \State \Call{QuadSearch}{$i_1, i_{mid}, j_{mid}+1, j_2$}
    \State \Call{QuadSearch}{$i_{mid}+1, i_2, j_{mid}+1, j_2$}
\EndProcedure
\end{algorithmic}
\end{algorithm}

The algorithm operates on the discrete set $W$ of unique capacity weights found in the graph: \textsc{SolveVPC} collects this set and invokes the recursive \textsc{QuadSearch} over the full index range, which subdivides each capacity block $[a_{min}, a_{max}] \times [b_{min}, b_{max}]$ into four quadrants by bisecting both index ranges. Degenerate blocks with an empty index range---which arise once a dimension has narrowed to a single index---are rejected upfront (Line 8), so blocks split safely even when only one dimension can still be subdivided. Two pruning strategies (detailed in Section \ref{subsec:pruning}) are applied at each recursive step to avoid redundant shortest-path computations. The algorithm terminates when all blocks have either been pruned or resolved to a plateau, at which point $\Phi^*$ holds the exact optimal objective value and $P^*$ the corresponding path.

\subsection{Pruning Strategies}
\label{subsec:pruning}

Although the capacity-space reformulation reduces the search from an exponential path enumeration to a finite grid, naively evaluating every $(a, b)$ pair is still prohibitive in dense networks with thousands of distinct edge weights. The quadtree algorithm achieves practical efficiency through two complementary pruning strategies that together allow large portions of the search space to be skipped without sacrificing exactness.

\subsubsection{Plateau Pruning}
\label{subsubsec:plateau}

The discrete nature of the underlying graph $G$ induces a ``plateau'' effect on the distance function $D(a, b)$: since the restricted subgraph $G_{a,b}$ only changes when an edge is added or removed, $D(a,b)$ remains constant over large contiguous regions of the capacity space.

\begin{figure*}[t]
    \centering
    \begin{subfigure}{0.16\textwidth}
        \centering
        \includegraphics[width=\textwidth]{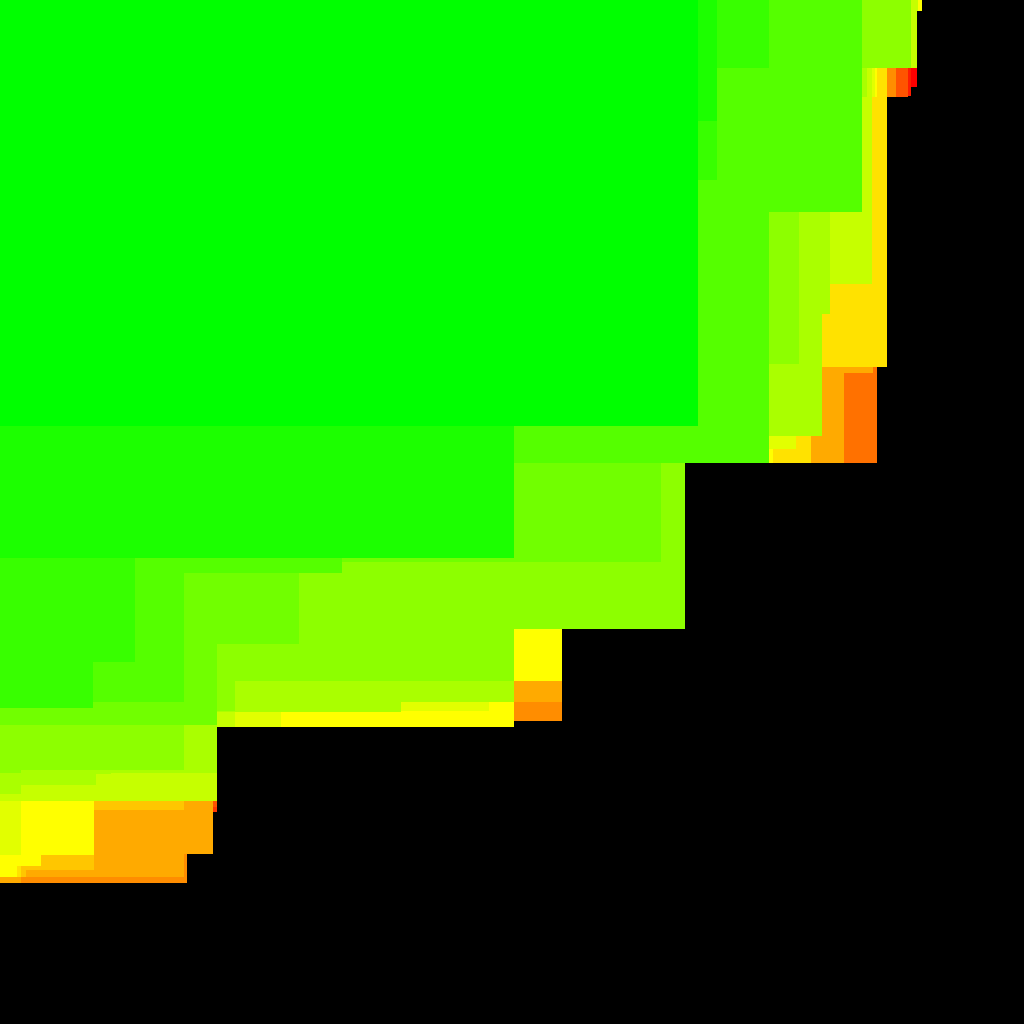}
        \caption{BA (CNIR, Pair A)}
    \end{subfigure}
    \begin{subfigure}{0.16\textwidth}
        \centering
        \includegraphics[width=\textwidth]{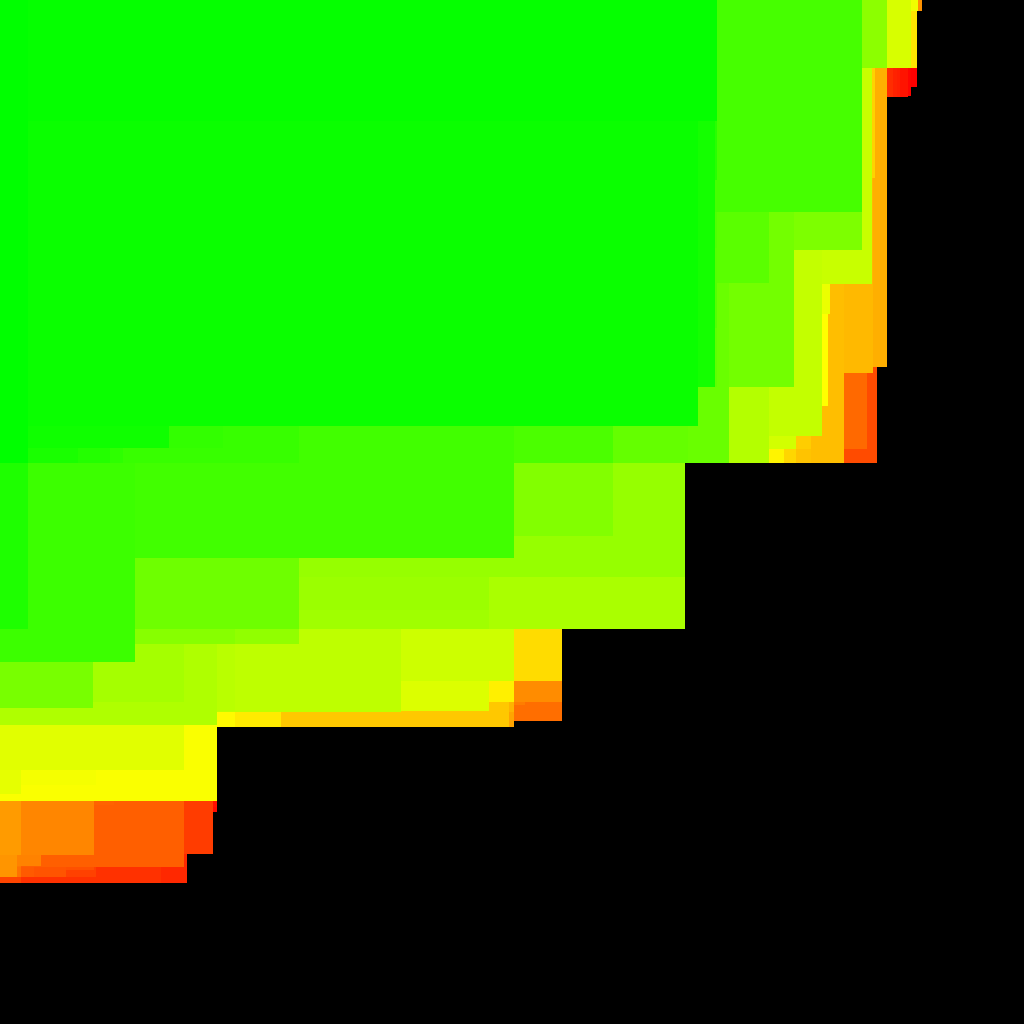}
        \caption{BA (SECER, Pair A)}
    \end{subfigure}
    \hfill
    \begin{subfigure}{0.16\textwidth}
        \centering
        \includegraphics[width=\textwidth]{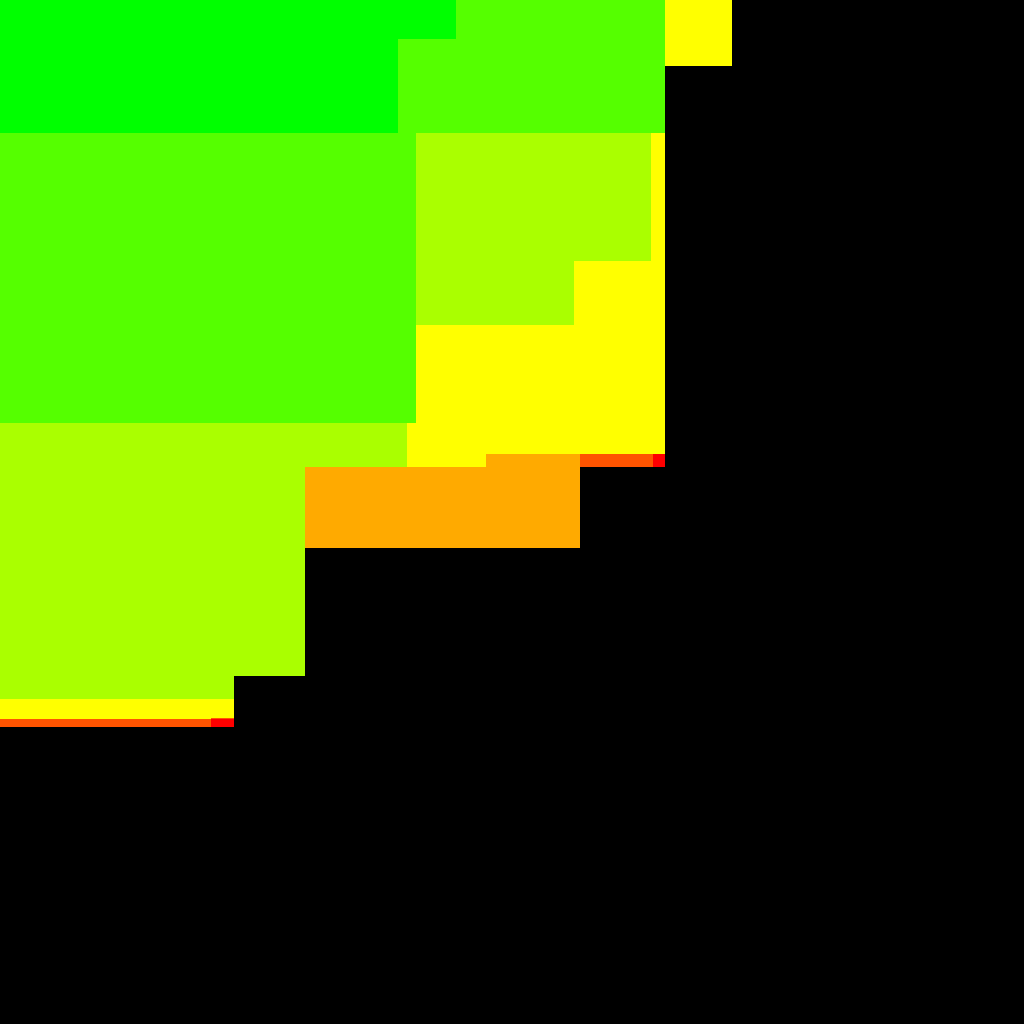}
        \caption{BA (CNIR, Pair B)}
    \end{subfigure}
    \begin{subfigure}{0.16\textwidth}
        \centering
        \includegraphics[width=\textwidth]{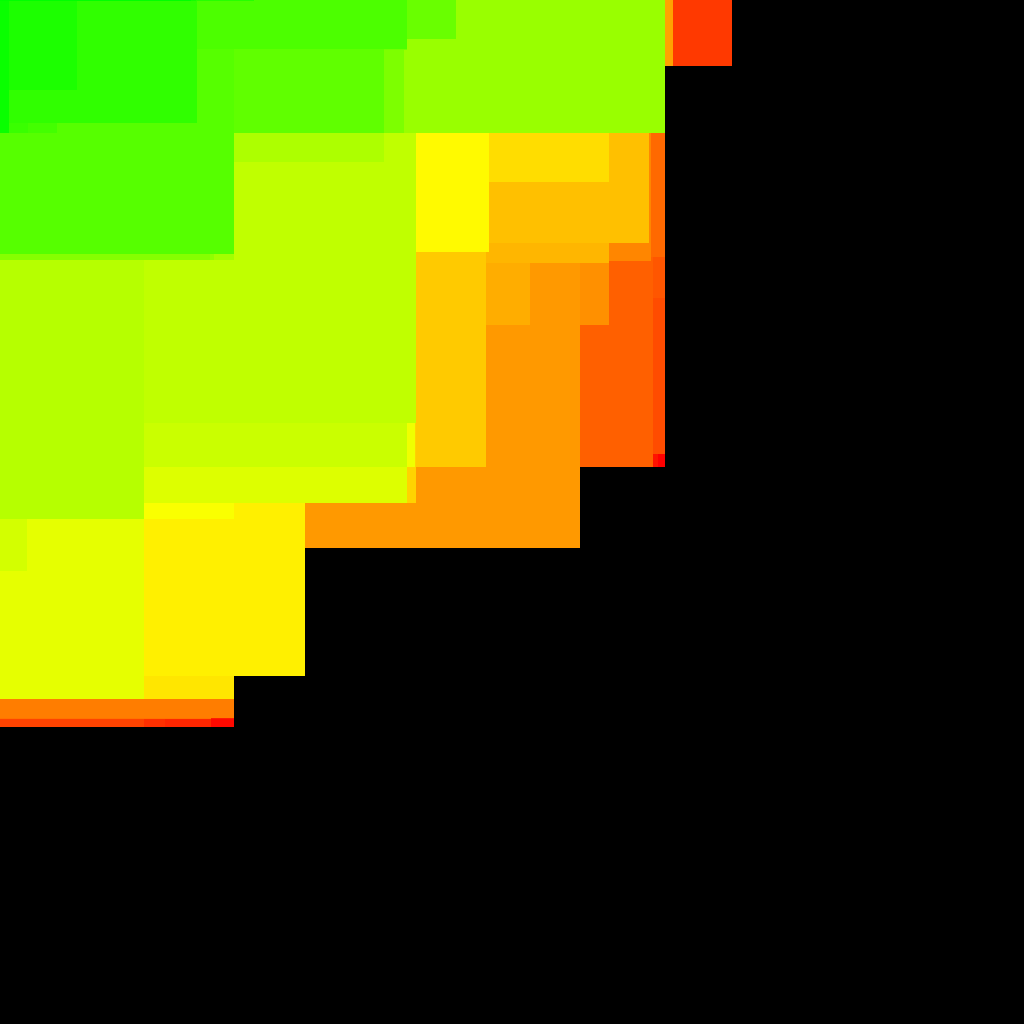}
        \caption{BA (SECER, Pair B)}
    \end{subfigure}
    \hfill
    \begin{subfigure}{0.16\textwidth}
        \centering
        \includegraphics[width=\textwidth]{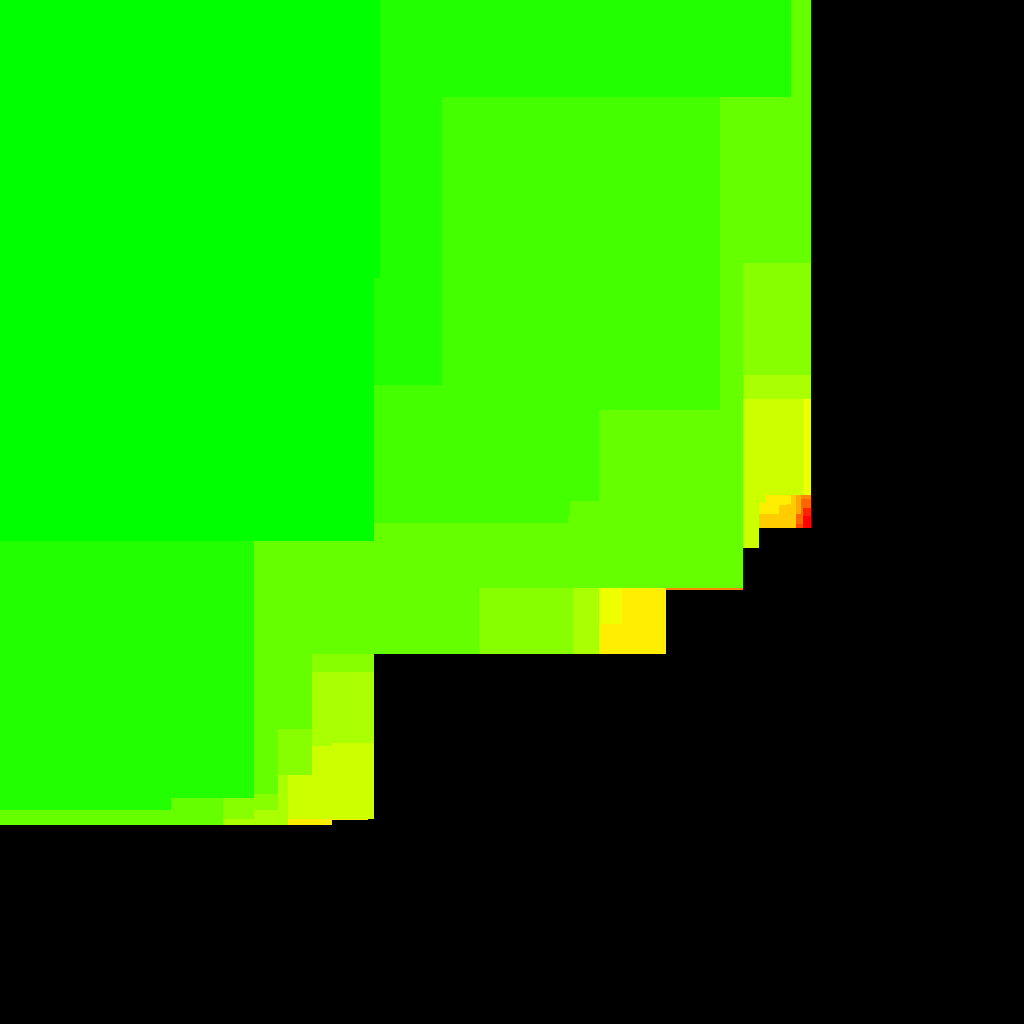}
        \caption{BA (CNIR, Pair C)}
    \end{subfigure}
    \begin{subfigure}{0.16\textwidth}
        \centering
        \includegraphics[width=\textwidth]{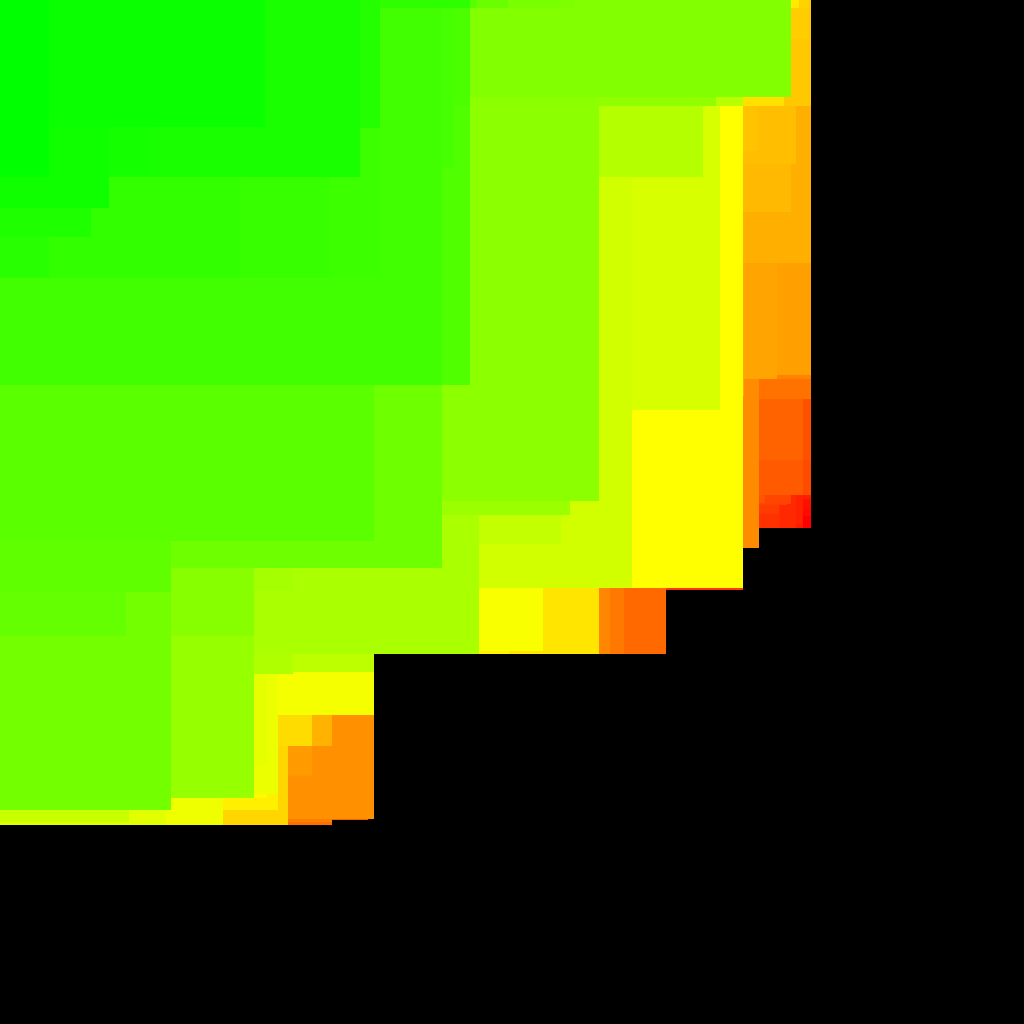}
        \caption{BA (SECER, Pair C)}
    \end{subfigure}

    \vspace{0.5em}

    \begin{subfigure}{0.16\textwidth}
        \centering
        \includegraphics[width=\textwidth]{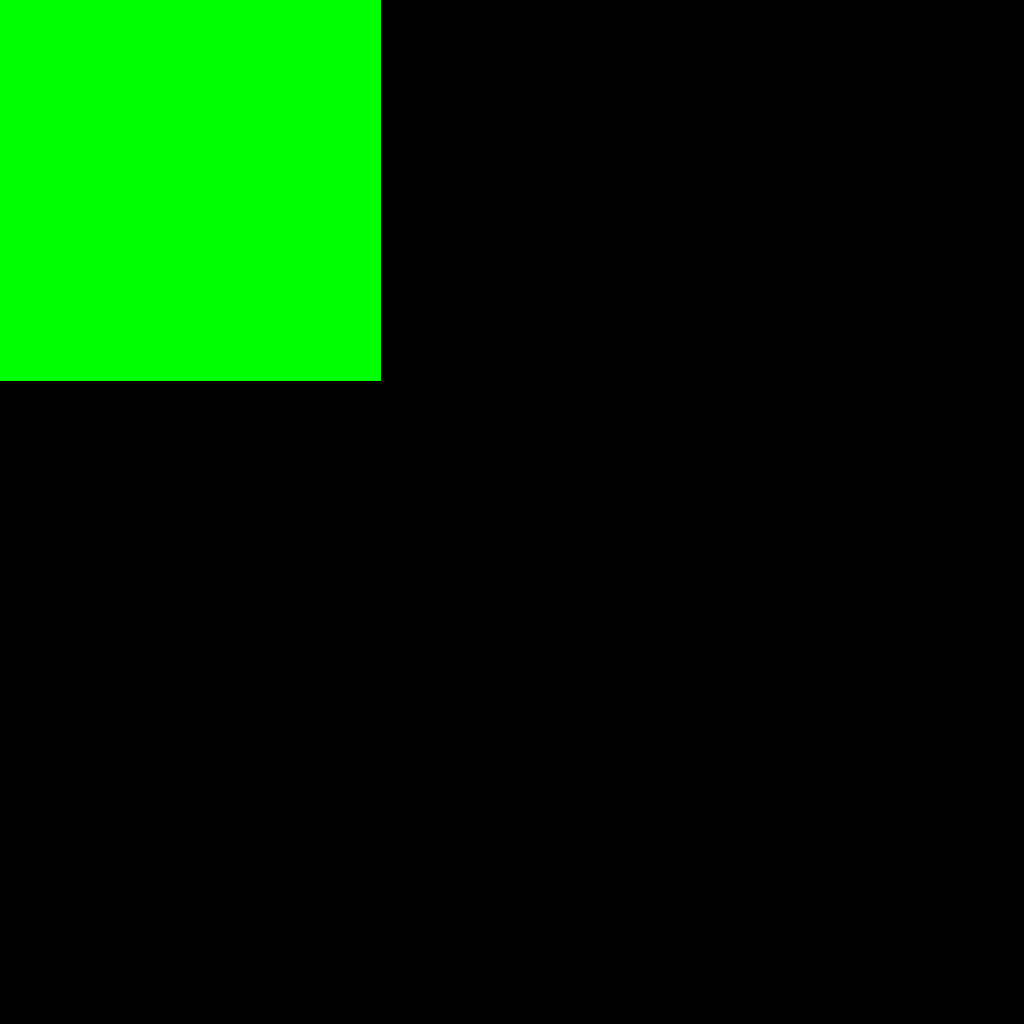}
        \caption{LN (CNIR, Pair D)}
    \end{subfigure}
    \begin{subfigure}{0.16\textwidth}
        \centering
        \includegraphics[width=\textwidth]{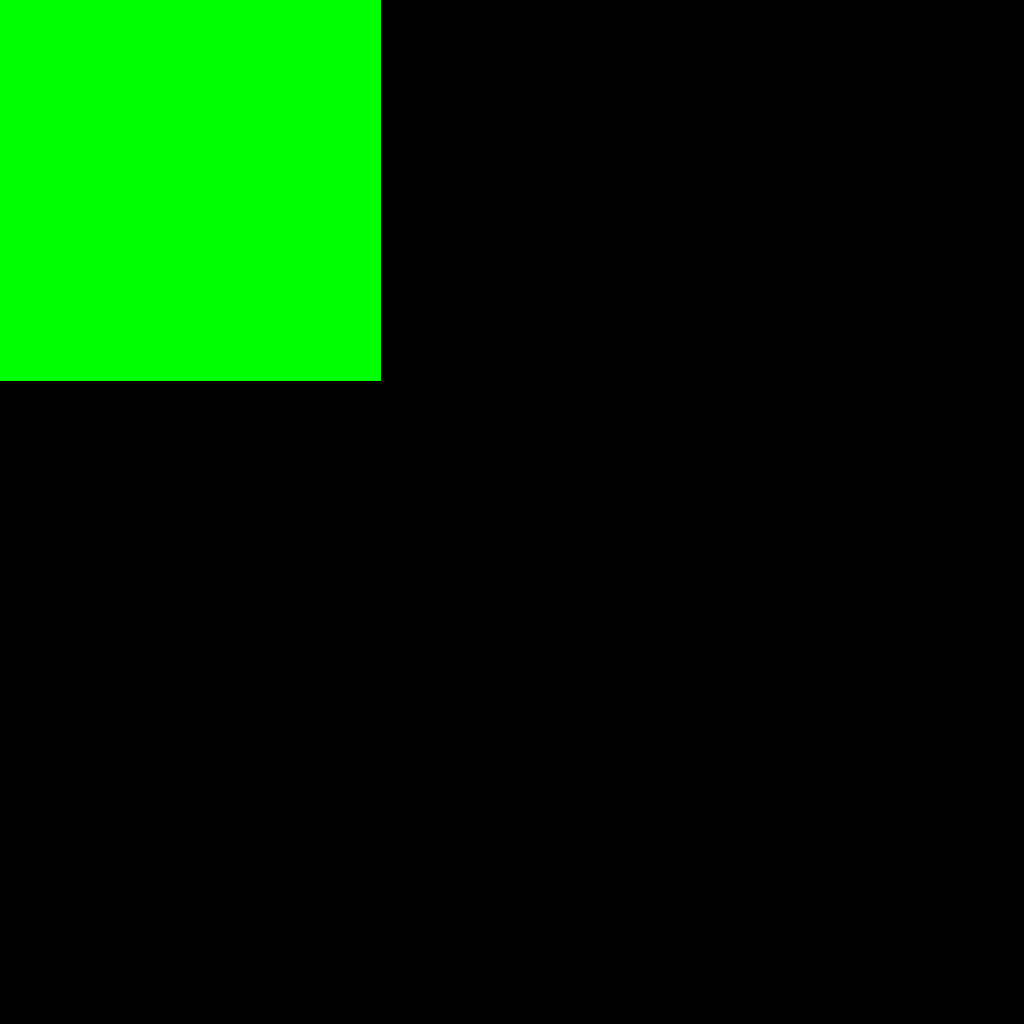}
        \caption{LN (SECER, Pair D)}
    \end{subfigure}
    \hfill
    \begin{subfigure}{0.16\textwidth}
        \centering
        \includegraphics[width=\textwidth]{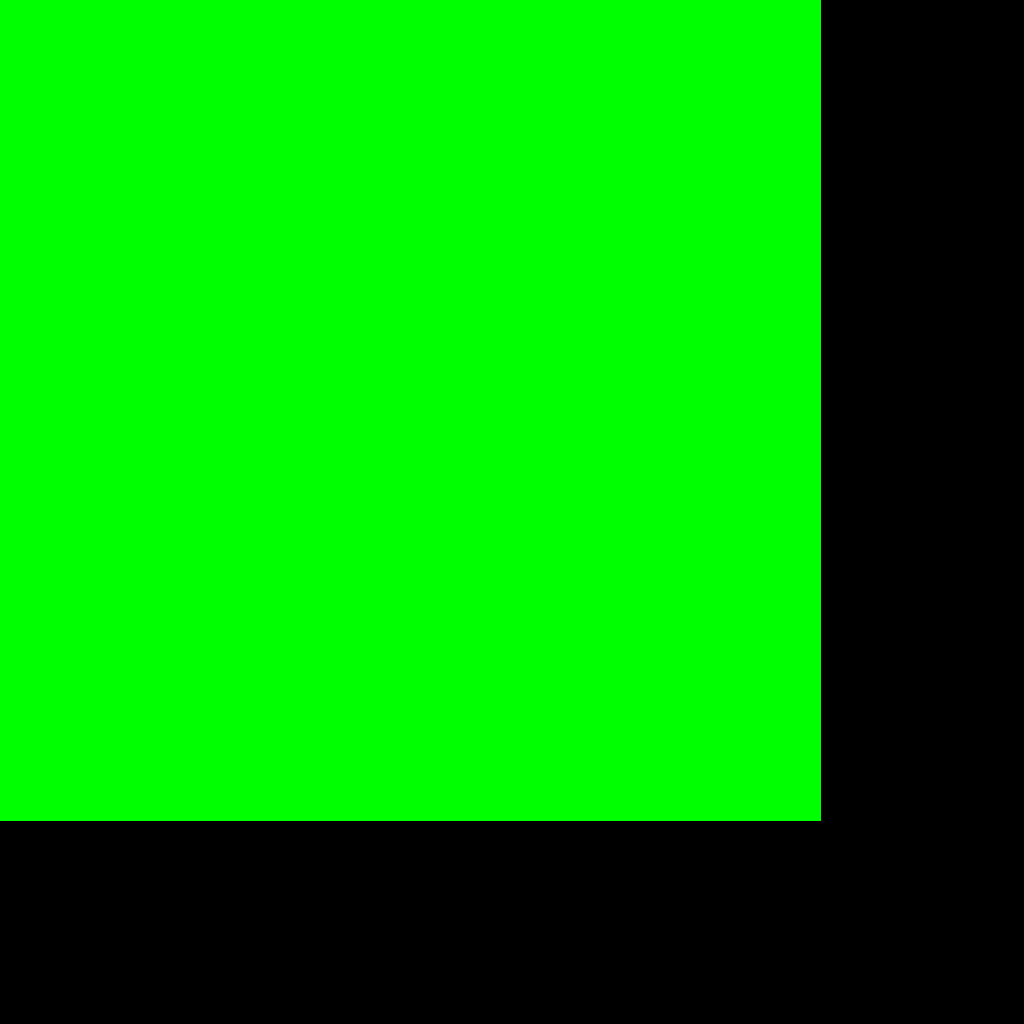}
        \caption{LN (CNIR, Pair E)}
    \end{subfigure}
    \begin{subfigure}{0.16\textwidth}
        \centering
        \includegraphics[width=\textwidth]{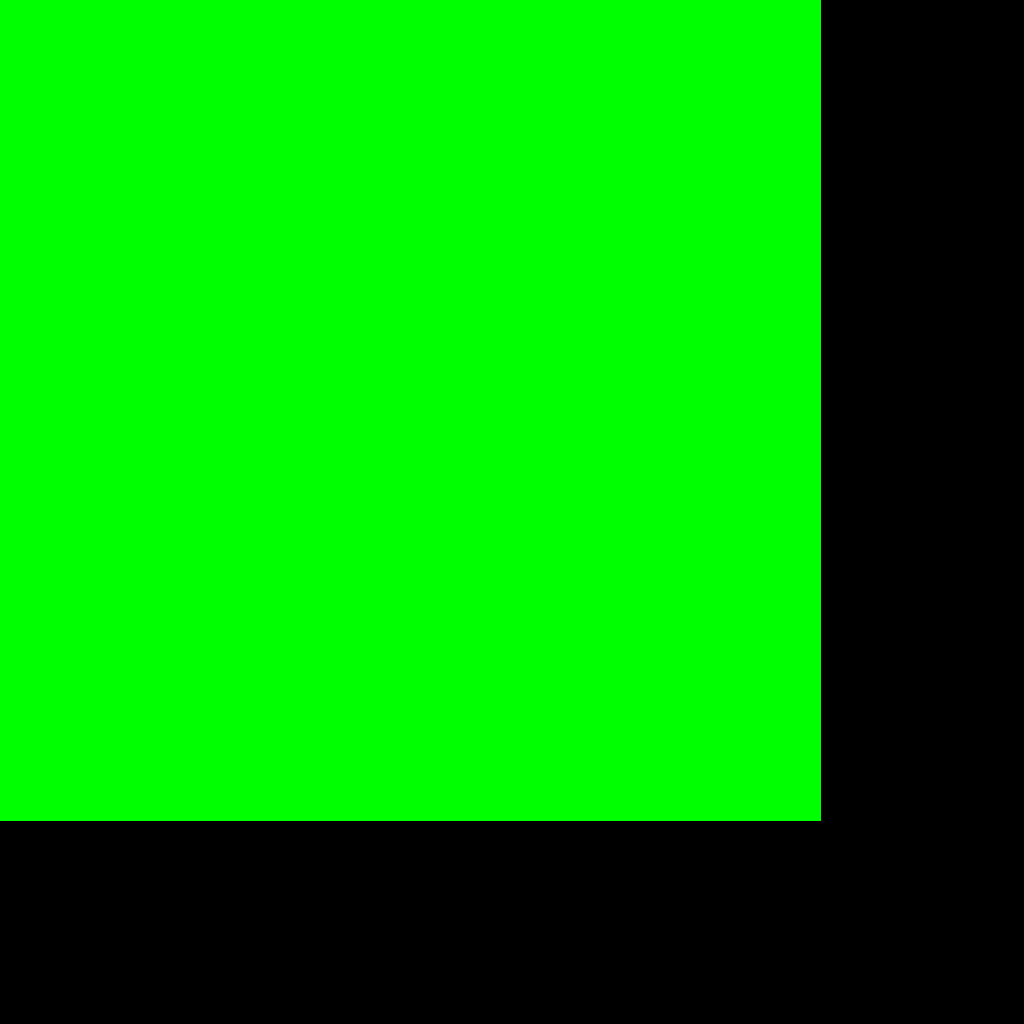}
        \caption{LN (SECER, Pair E)}
    \end{subfigure}
    \hfill
    \begin{subfigure}{0.16\textwidth}
        \centering
        \includegraphics[width=\textwidth]{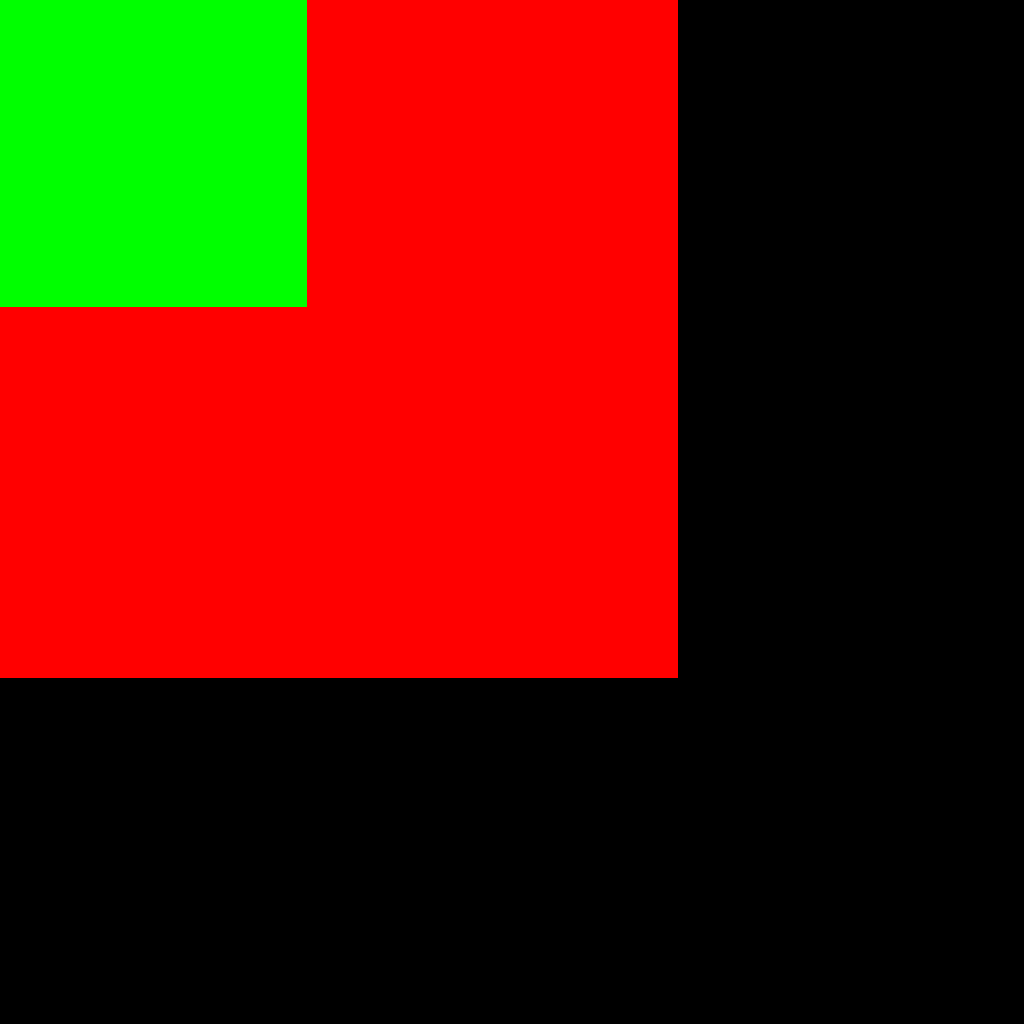}
        \caption{LN (CNIR, Pair F)}
    \end{subfigure}
    \begin{subfigure}{0.16\textwidth}
        \centering
        \includegraphics[width=\textwidth]{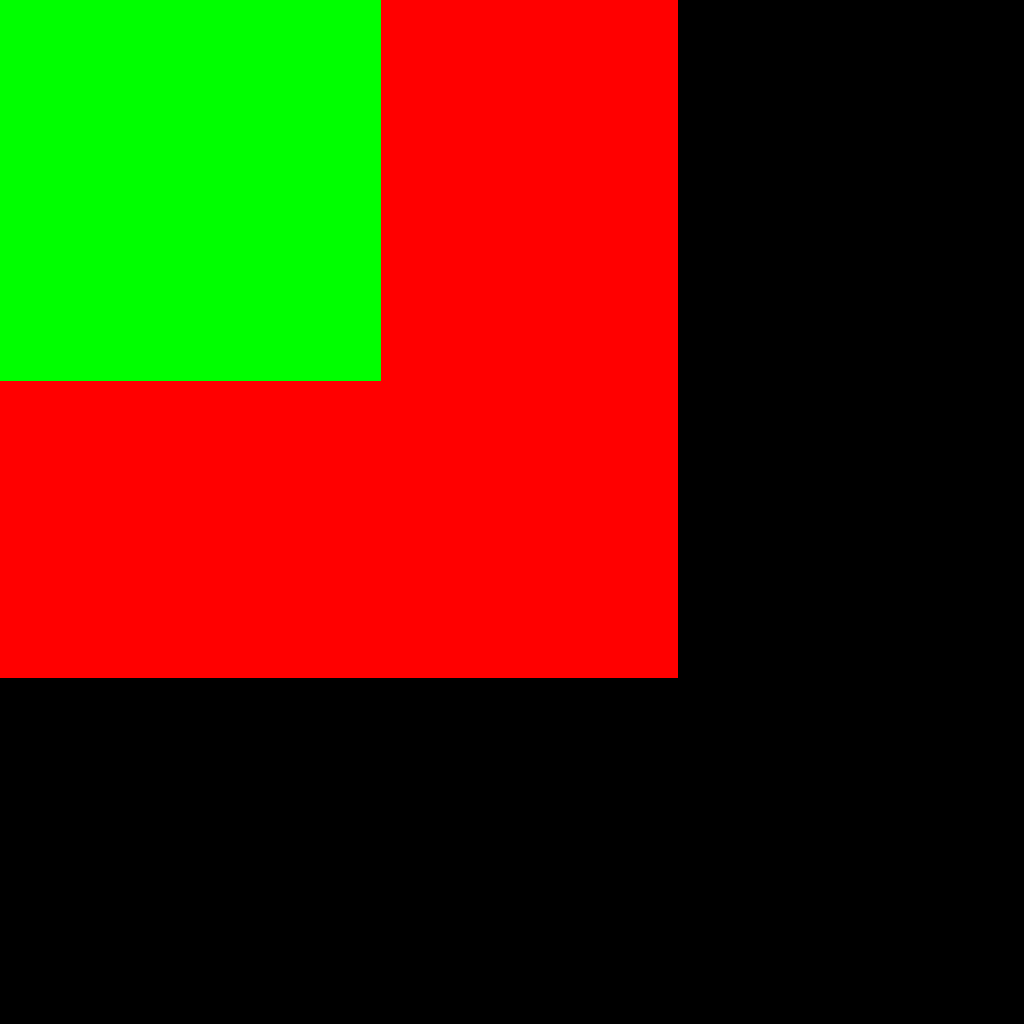}
        \caption{LN (SECER, Pair F)}
    \end{subfigure}
    \caption{Visualization of the distance function $D(a, b)$ across different topologies and metrics. Top: BA random graphs. Bottom: Lightning Network topology. For each instance, we show $D(a, b)$ for a specific source-target pair under both the CNIR metric (left) and the more complex cost metric detailed in Section \ref{sec:metric} (right). The color represents the shortest distance value, with black regions indicating infinite distance (i.e., no feasible path exists for the given capacity thresholds), demonstrating that the large-scale ``plateau'' structure is a universal phenomenon exploited by our solver.}
    \Description{Twelve heatmaps of the distance function over the two-dimensional capacity space, arranged as CNIR and SECER pairs for three source-target pairs on BA graphs in the top row and three on the Lightning Network in the bottom row. Every heatmap consists of large contiguous constant-color plateau regions; the Lightning Network instances contain only a few large plateaus, while the BA graphs show more granular staircase-like patterns.}
    \label{fig:heatmaps}
\end{figure*}

Figure \ref{fig:heatmaps} visualizes this effect, showing $D(a, b)$ across the capacity space for various source-target pairs under both the standard CNIR metric and the economic-aware metric of Section \ref{sec:metric}. Plateaus emerge universally across topologies and cost metrics: Lightning Network instances (bottom row) are remarkably sparse, often consisting of a few large constant-distance plateaus, while synthetic BA graphs (top row) exhibit more granular patterns with frequent transitions. In all cases, large contiguous plateaus exist, indicating that the network topology itself---rather than the specific path metric---primarily drives the spatial structure of $D(a, b)$ and directly validating our plateau pruning strategy.

The following results formalize this structure and show how it can be exploited for pruning.

\begin{lemma}[Distance Plateau]
    \label{lemma:plateau}
    For any rectangular region $R = [a_1, a_2] \times [b_1, b_2]$ in the capacity space, if $D(a_1, b_1) = D(a_2, b_2) = C$, then for all $(a, b) \in R$, $D(a, b) = C$.
\end{lemma}
\begin{proof}
    By Lemma \ref{lemma:mono}, $D(a_1, b_1) \le D(a, b) \le D(a_2, b_2)$ for any $(a, b) \in R$. Since $D(a_1, b_1) = D(a_2, b_2) = C$, it follows that $D(a, b) = C$.
\end{proof}

\begin{theorem}[Plateau Pruning]
    \label{thm:sparsity}
    If a rectangular region $R = [a_1, a_2] \times [b_1, b_2]$ satisfies $D(a_1, b_1) = D(a_2, b_2)$, then the objective function $\Phi(a, b) = \frac{a+b}{D(a, b)}$ within $R$ is maximized at the maximal corner $(a_2, b_2)$.
\end{theorem}
\begin{proof}
    By Lemma \ref{lemma:plateau}, $D(a, b) = C$ for all $(a, b) \in R$. Thus $\Phi(a, b) = \frac{a+b}{C} \le \frac{a_2+b_2}{C} = \Phi(a_2, b_2)$.
\end{proof}


In the algorithm, after computing the distances at both diagonal corners of a block---$d_{min} = D(a_{min}, b_{min})$ and $d_{max} = D(a_{max}, b_{max})$ (Lines 14 and 18)---if $d_{min} = d_{max}$, Theorem \ref{thm:sparsity} guarantees that the block's optimal solution is at the maximal corner $(a_{max}, b_{max})$. The algorithm then updates $\Phi^*$ and returns immediately (Lines 19--25), skipping all interior points and sub-blocks. This pruning mechanism is particularly effective when the search space contains large regions of uniform distances.

\subsubsection{Threshold Pruning}
\label{subsubsec:threshold}

The second strategy exploits the structure of the objective function to derive an upper bound on the path distance that any candidate solution within a block can usefully have. For a block $[a_{min}, a_{max}] \times [b_{min}, b_{max}]$, the numerator $a + b$ is bounded above by $a_{max} + b_{max}$. Therefore, for any $(a, b)$ in the block to improve on the current best $\Phi^*$, its path distance $d = D(a, b)$ must satisfy:
\[
    \frac{a_{max} + b_{max}}{d} > \Phi^* \Rightarrow d < \frac{a_{max} + b_{max}}{\Phi^*}.
\]
We therefore define the \emph{depth limit}
\[
    d_{lim} = \begin{cases} \infty & \text{if } \Phi^* = 0, \\ \frac{a_{max} + b_{max}}{\Phi^*} & \text{if } \Phi^* > 0, \end{cases}
\]
and introduce the truncated distance function:
\[
    D_{d_{lim}}(a, b) = \begin{cases} D(a, b) & \text{if } D(a, b) \le d_{lim}, \\ \infty & \text{if } D(a, b) > d_{lim}. \end{cases}
\]
The following theorem justifies replacing $D$ with $D_{d_{lim}}$ during the search without loss of optimality.

\begin{theorem}[Correctness of Threshold Pruning]
    \label{thm:threshold}
    Let $\Phi^*$ be the current best objective value and $d_{lim} = \frac{a_{max}+b_{max}}{\Phi^*}$ for a block $R = [a_{min}, a_{max}] \times [b_{min}, b_{max}]$. No point $(a, b) \in R$ with $D(a, b) > d_{lim}$ can improve $\Phi^*$.
\end{theorem}
\begin{proof}
    For any $(a, b) \in R$ with $D(a, b) > d_{lim}$, we have $\Phi(a, b) = \frac{a + b}{D(a, b)} \le \frac{a_{max} + b_{max}}{D(a, b)} < \frac{a_{max} + b_{max}}{d_{lim}} = \Phi^*$.
\end{proof}

In the algorithm, $d_{lim}$ is passed to the shortest-path function as an early-termination threshold (Lines 11--14, 18): the search is aborted as soon as the tentative distance exceeds $d_{lim}$, returning $\infty$ immediately. The check is applied first to the most permissive constraint $(a_{min}, b_{min})$, since by Lemma \ref{lemma:mono} this yields the smallest possible distance in the block. If even this relaxed search returns $\infty$ (Lines 15--17), then the entire block is guaranteed to be $\infty$ and is pruned.

Furthermore, this thresholding mechanism provides a significant secondary benefit by simplifying the distance landscape. By mapping all distances exceeding $d_{lim}$ to $\infty$, the search space becomes more uniform in regions with low objective values. This allows the Plateau Pruning strategy (Section \ref{subsubsec:plateau}) to effectively identify and skip these regions much earlier in the quadtree decomposition. As the algorithm discovers better solutions and $\Phi^*$ increases, $d_{lim}$ becomes more restrictive, causing more blocks to be discarded or combined into these uniform infinite-distance regions.

Since both pruning rules are individually safe, the exactness of the full procedure follows directly:

\begin{corollary}[Exactness of \textsc{SolveVPC}]
    \label{cor:exactness}
    Algorithm \ref{alg:quadtree} returns $\Phi^* = \max_{P \in \mathcal{P}_{s\to t}} \Phi(P)$.
\end{corollary}
\begin{proof}
    The recursion covers the entire index grid, and a block is discarded only by Theorem \ref{thm:sparsity} or Theorem \ref{thm:threshold}, neither of which can eliminate a point that improves on the incumbent $\Phi^*$. Every surviving point is eventually resolved at a plateau corner, so upon termination $\Phi^* = \max_{a, b \in W} \frac{a+b}{D(a,b)}$, which equals $\max_{P} \Phi(P)$ by Theorem \ref{thm:equivalence}.
\end{proof}

\subsection{Complexity Analysis}
\label{subsec:complexity}

Let $n$ and $m$ denote the number of nodes and edges in the PCN, respectively. The discrete search space for the capacity thresholds $(a, b)$ is bounded by $|W| \times |W| = O(m^2)$. Without any pruning, the solver would resort to an exhaustive evaluation of all capacity pairs, each requiring an $O(n + m)$ Dijkstra computation, yielding a theoretical worst-case complexity of $O(m^2(n + m))$. This bound matches the baseline established by Thor \cite{wei2024thor}. The space overhead beyond the graph itself is negligible, as the recursion depth is bounded by $O(\log |W|)$.

\begin{figure*}[t]
    \centering
    \begin{subfigure}[c]{0.83\linewidth}
        \centering
        \includegraphics[width=\linewidth]{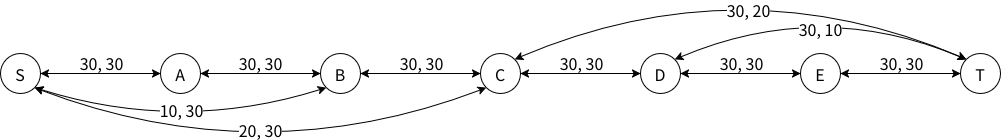}
        \caption{Adversarial structure}
        \label{fig:adv_graph}
    \end{subfigure}
    \hfill
    \begin{subfigure}[c]{0.15\linewidth}
        \centering
        \includegraphics[width=\linewidth]{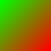}
        \caption{Distance function $D(a, b)$}
        \label{fig:adv_heatmap}
    \end{subfigure}
    \caption{Structure of the adversarial graph in small scale and the resulting distance function landscape in large scale.}
    \Description{Left: an adversarial graph consisting of a backbone path between the source and target, coupled with staggered bypass edges of increasing capacities. Right: the resulting distance function heatmap over the capacity space, showing a fine staircase pattern with almost no large plateaus.}
    \label{fig:adv_structure}
\end{figure*}

To investigate the limits of the quadtree decomposition, we introduce a class of \emph{adversarial graphs} specifically designed to minimize the occurrence of distance plateaus (see Figure \ref{fig:adv_structure}). By coupling a backbone path with staggered bypass capacities, we force the distance function $D(a, b)$ to change at nearly every increment of $a$ or $b$, rendering the sparsity-based pruning (Section \ref{subsubsec:plateau}) alone insufficient and potentially forcing the algorithm toward its theoretical complexity limit.

In practice, however, the two pruning mechanisms act \textbf{synergistically}. As the search progresses and the incumbent $\Phi^*$ improves, the dynamic depth limit $d_{lim}$ becomes increasingly restrictive, truncated regions collapse into uniform $\infty$-plateaus, and Plateau Pruning skips them wholesale---even when the distance function is strictly monotonic. The effective running time is therefore \emph{output-sensitive}: computation concentrates at the phase boundaries of $D(a, b)$ where the shortest path actually changes, rather than on the full $O(m^2)$ grid. In topologies like the Lightning Network, whose discrete capacity distributions induce large contiguous plateaus, most of the space is resolved with a handful of Dijkstra evaluations. Our experiments (Section \ref{sec:eval:efficiency}) confirm this analysis: the solver maintains an execution time below 100\,ms even on adversarial instances with 12,000 nodes, achieving a 2--5 orders of magnitude speedup over the Thor baseline and demonstrating that the combined pruning strategies collapse the search space even in worst-case configurations.

Finally, we remark that the four sub-blocks generated at each recursion step are mutually independent, so the quadtree search decomposes naturally into a task-parallel tree traversal. Our implementation exploits this at the thread level (Section \ref{sec:evaluation}); extending the same decomposition to distributed or GPU execution is a promising direction for future work.

\section{The SECER Metric}
\label{sec:metric}

While the quadtree-based solver developed in Section~\ref{sec:solver} addresses the computational challenges of VPC path selection, its practical effectiveness is contingent upon a distance function $D(P)$ that accurately captures the multifaceted costs of routing. Traditional metrics, such as the hop-based CNIR, often oversimplify the network by assuming homogeneous channel characteristics and ignoring the inherent risks of off-chain transactions. 

Balancing multiple performance indicators---capacity, transaction fees, and security risks---is fundamentally a Multi-Objective Optimization (MOO) problem \cite{verma2021comprehensive,salzman2023heuristic,warburton1987approximation}. Rather than enumerating the exact Pareto front of non-dominated paths as in EMOA* \cite{ren2025emoa}, or analyzing weighted combinations via linear-parametric optimization theory \cite{nemesch2025survey}, our framework treats the core optimization algorithm as a strictly high-performance shortest-path engine over a capacity-constrained space, and leaves it to the \textit{user} to fuse these orthogonal objectives into a single monotonic distance metric $D(P)$ representing their preferences. The solver's objective is not to map the entire multi-dimensional space, but to rapidly find the optimal path for a specific, user-defined preference.

To demonstrate this, we propose the \textit{Secure and Economic-aware Channel Efficiency Ratio} (SECER), which unifies two critical dimensions of the Lightning Network---economic overhead and security risk---by instantiating the generalized distance function as a weighted sum:
\begin{equation}
    D_{\text{SECER}}(P) = \alpha \cdot Cost(P) + \beta \cdot Risk(P)
\end{equation}
where $\alpha \ge 0$ and $\beta \ge 0$ are user-defined preference weights that reflect the relative importance of economy versus security. The resulting objective function for SECER is defined as:
\begin{equation}
    \Phi_{\text{SECER}}(P) = \frac{\delta(P^{s\to t}) + \delta(P^{t\to s})}{D_{\text{SECER}}(P)}
\end{equation}
This formulation allows the VPC path selection to adapt to diverse user requirements. By adjusting $\alpha$ and $\beta$, users can transition smoothly from a "cost-optimized" path for micro-payments to a "security-hardened" path for significant asset transfers. In the following, we detail the modeling of each component.

\subsection{Economic Cost Modeling}
Given the absence of large-scale VPC deployments, standardized fee structures remain an open area of research \cite{aumayr2025optimizing, xu2025srvc}. This work adopts the recursive fee model proposed by Aumayr et al. \cite{aumayr2025optimizing} as a representative economic benchmark for path evaluation; the fee model itself is thus not a contribution of this work---SECER's novelty lies in the continuous risk model and the tunable fusion that our exact solver optimizes directly. For a multi-hop path $P = (v_0, v_1, \dots, v_n)$, we define the total economic cost $Cost(P)$ as the sum of cumulative fees charged by all intermediaries:
\begin{equation}
    Cost(P) = \sum_{j=1}^{n-1} f_{j, j+1}
\end{equation}
where the fee $f_{j, j+1}$ charged by intermediary $v_j$ is defined as:
\begin{equation}
    \label{eq:aumayr_cost}
    f_{j, j+1} = \text{base\_fee}_{j, j+1} + \text{prop\_fee}_{j, j+1} \cdot (A + \sum_{k=j+1}^{n-1} f_{k, k+1})
\end{equation}
The term $A$ denotes a nominal target capacity specified by the user at initiation and is treated as a pre-determined constant throughout the optimization. Consequently, $D(P)$ depends solely on the static weights of the path $P$ and remains independent of the capacity thresholds $(a, b)$ explored by the solver---precisely the stability condition required in Section \ref{sec:solver} to preserve the distance plateau property. The billing rules themselves are replaceable: any fee model that preserves monotonicity can be adopted for future protocol specifications without compromising the solver's correctness.

\subsection{Security Risk Assessment}
Beyond economic considerations, security remains a paramount concern in Layer-2 solutions, as VPCs are fundamentally challenged by on-path adversaries capable of performing value-privacy or wormhole attacks. While prior work \cite{aumayr2021bitcoin} utilizes a binary bypass strategy---categorizing nodes as either secure or corrupted---our framework introduces a continuous risk assessment model. This approach allows the solver to navigate the spectrum of adversarial threats rather than employing rigid exclusions.

We adopt the adversarial model of an adaptive attacker with a fixed budget who seeks to maximize disruption by corrupting high-impact nodes. The strategic utility of a node $v$ to such an adversary is quantified via a \textit{Cost-Benefit Ratio} ($CB$):
\begin{equation}
    CB(v) = \frac{occ(v) / N_{\text{pay}}}{C_{\text{locked}}(v) / \mathcal{B}_{\text{adv}}}
\end{equation}
where $occ(v)$ represents the node's traffic hit count (the frequency of $v$ appearing in a set of $N_{\text{pay}}$ simulated payment paths), and $C_{\text{locked}}(v)$ denotes the total capacity locked by $v$ across its incident channels. The term $\mathcal{B}_{\text{adv}}$ represents the total adversarial budget. This ratio captures the trade-off between the attack's coverage (benefit) and the financial cost of corruption (cost).

The cumulative risk of a path $P$ is modeled as the sum of the cost-benefit ratios of its intermediate nodes:
\begin{equation}
    Risk(P) = \sum_{v \in P \setminus \{s, t\}} CB(v)
\end{equation}
This additive formulation captures the expanding attack surface associated with additional intermediaries. By integrating these strategic scores directly into the SECER objective, the optimization process can quantitatively weigh the adversarial attractiveness of a path against its economic costs.

Note that, while we currently utilize a summation-based risk model, alternative formulations---such as bottleneck risk (the maximum risk of any single node) or non-linear probability-based risk---can be substituted, as long as they maintain monotonicity and can be efficiently computed within the shortest-path sub-routine.

\section{Evaluation}
\label{sec:evaluation}

We evaluate our framework along two primary dimensions: the computational performance of the quadtree-based solver and the qualitative effectiveness of the SECER metric, comparing against prior baselines across diverse network topologies.

\subsection{Experimental Setup}

All experiments are conducted on a PC equipped with an AMD Ryzen 9 9900X CPU and 64GB of DDR5 memory. The core optimization framework and the quadtree-based solver are implemented in Rust to ensure high performance and memory safety. For a fair, apples-to-apples comparison, we also re-implemented the Thor Ag1 baseline in Rust using identical shortest-path subroutines and data structures. Our source code and the Lightning Network datasets are publicly available on GitHub\footnote{\url{https://github.com/swwind/quadtree-vpc}}.
We utilize the Rayon library for parallel computing to exploit the multi-core architecture of our hardware. Unless otherwise noted, all experiments are executed using 24 parallel threads.

\subsection{Algorithmic Efficiency}
\label{sec:eval:efficiency}

\begin{figure*}[t]
  \centering
  \includegraphics{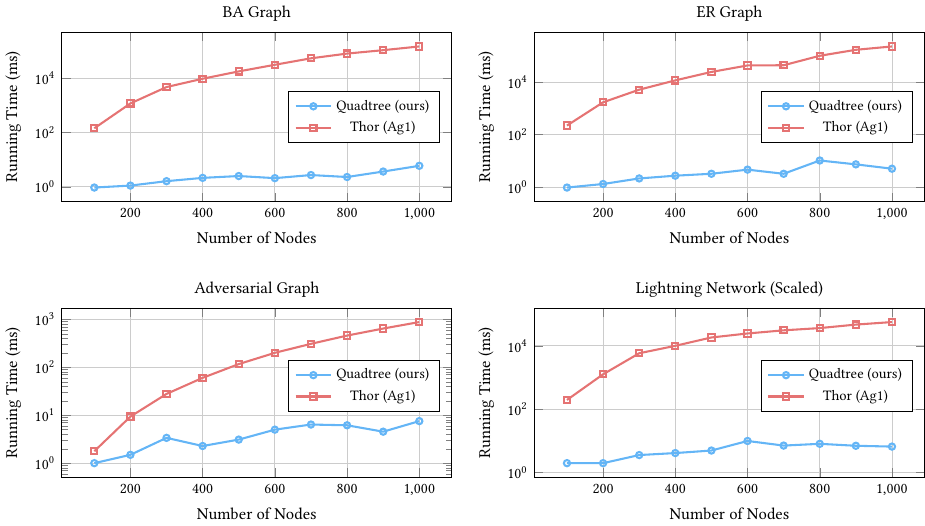}
  \caption{Running time (log scale) of Quadtree vs. Thor (Ag1) on four graph types.
    Our solver achieves 2--5 orders of magnitude speedup across all settings.}
  \Description{Four log-scale line charts of running time versus number of nodes for the Quadtree solver and Thor Ag1 on BA, ER, adversarial, and Lightning Network graphs. In all four charts the Quadtree curve lies two to five orders of magnitude below the Thor Ag1 curve.}
  \label{fig:exp1:comparison}
\end{figure*}

\begin{figure}[t]
  \centering
  \includegraphics{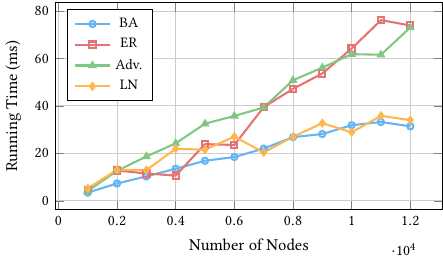}
  \caption{Scalability of the Quadtree solver on larger graphs ($n$ up to 12,000).
    Thor (Ag1) is omitted as it is computationally infeasible at this scale.}
  \Description{Line chart of Quadtree solver running time on graphs of 1,000 to 12,000 nodes. All curves remain below 80 milliseconds, with BA and Lightning Network graphs below 40 milliseconds and ER and adversarial graphs slightly higher.}
  \label{fig:exp1:large}
\end{figure}

We compare our quadtree solver with Thor's Ag1 baseline across four graph types: (i) Barabási--Albert (BA) graphs \cite{albert2002statistical}, modeling power-law degree distributions; (ii) Erdős--Rényi (ER) graphs \cite{erdHos1960evolution}, representing uniform random connectivity; (iii) adversarial graphs, constructed to minimize distance plateaus and hinder pruning (Figure~\ref{fig:adv_structure}); and (iv) Lightning Network (LN) snapshots from \texttt{1ml.com}, containing 12,552 nodes and 43,936 channels. Synthetic graphs use a consistent average degree $k=4$ to match LN density, except for the adversarial case (Section~\ref{subsec:complexity}).

As shown in Figure~\ref{fig:exp1:comparison}, our method is consistently 2--5 orders of magnitude faster than Ag1. For example, on a BA graph with 1,000 nodes, Ag1 takes up to 149s, while our solver completes in under 10ms.

To evaluate scalability, we run the quadtree solver on larger graphs with $n$ from 1,000 to 12,000 (Figure~\ref{fig:exp1:large}). The runtime remains below 80ms across all graph types. BA and LN graphs show similar performance (within 40ms), while ER and adversarial graphs are slightly slower but still under 80ms. This aligns with prior work indicating that LN exhibits scale-free properties similar to BA graphs \cite{seres2020topological,martinazzi2019evolution,zhao2021temporal}, which favor our pruning strategies by leveraging high-degree hub nodes to compress the search space.

Notably, despite being designed to suppress distance plateaus and challenge Plateau Pruning (Section~\ref{subsec:complexity}), the adversarial graph still yields runtimes under 80ms at $n=12,000$. This validates our analysis: the combination of Plateau and Threshold Pruning effectively reduces the search space even without large plateaus, ensuring robust performance in both typical and worst-case topologies.

\subsection{Ablation Study}
\label{sec:eval:ablation}

\begin{figure}[t]
  \centering
  \includegraphics{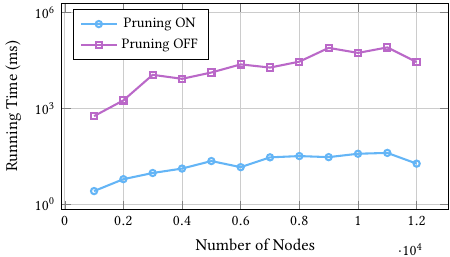}
  \caption{Running time comparison with and without pruning on the Lightning Network. Pruning provides a 3--4 orders of magnitude speedup.}
  \Description{Log-scale line chart of running time versus number of nodes on the Lightning Network with pruning enabled and disabled. The pruning-off curve grows rapidly beyond 80 seconds while the pruning-on curve stays below 50 milliseconds.}
  \label{fig:exp2:time}
\end{figure}

\begin{figure}[t]
  \centering
  \includegraphics{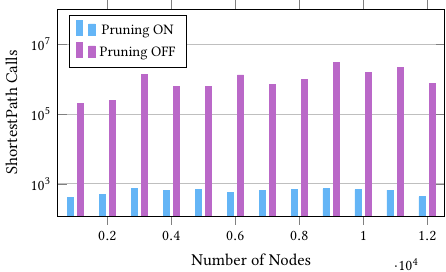}
  \caption{Number of \texttt{ShortestPath} calls with and without pruning. Our strategies reduce the computational load by several orders of magnitude.}
  \Description{Log-scale line chart of the number of ShortestPath calls versus number of nodes with pruning enabled and disabled. Pruning reduces the number of calls from over one million to under one thousand.}
  \label{fig:exp2:calls}
\end{figure}

To evaluate the effectiveness of our pruning mechanisms, we conduct an ablation study on the Lightning Network graph. Plateau pruning cannot be meaningfully ablated in isolation: the plateau test is the quadtree's recursion-termination condition, and removing it degenerates the search into the full-grid enumeration already represented by the Thor baseline in Section \ref{sec:eval:efficiency}. We therefore compare the full solver (Pruning ON) against a variant with threshold pruning disabled (Pruning OFF), isolating the contribution of threshold pruning and the synergy between the two strategies.

As shown in Figure~\ref{fig:exp2:time}, the synergistic pruning strategies achieve a \textbf{3--4 orders of magnitude} reduction in running time on the Lightning Network. Without the threshold pruning, the solver's execution time grows rapidly, exceeding $80$s at $n=11,000$, whereas the optimized solver consistently finishes in under $50$ms.

Figure~\ref{fig:exp2:calls} reflects the effectiveness of our pruning from a different perspective. Our mechanism not only reduces the average search depth within each \texttt{ShortestPath} call by early termination, but also eliminates a vast number of redundant calls by merging plateaus that cannot possibly contain a better solution. Specifically, the number of calls is lowered from over $10^6$ to under $10^3$, demonstrating that the synergistic pruning effectively collapses the search space and significantly reduces the overall computational overhead.

\subsection{The Cost-Security Trade-off}
\label{sec:eval:tradeoff}

\begin{figure}[t]
  \centering
  \includegraphics{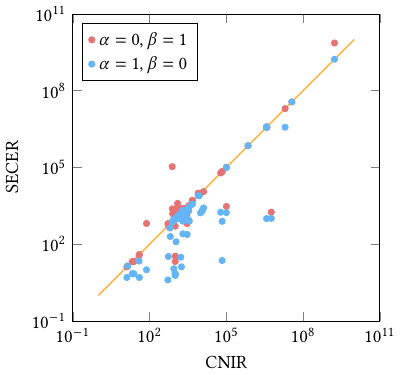}
  \caption{Comparison of economic cost between CNIR and SECER. SECER can substantially reduce economic cost relative to CNIR when $\alpha=1,\beta=0$ (blue marks).}
  \Description{Scatter plot comparing the economic cost of paths selected by CNIR on the horizontal axis and SECER on the vertical axis. In cost-optimized mode all points lie on or below the diagonal, indicating SECER selects paths with equal or lower fees.}
  \label{fig:exp3:costs}
\end{figure}

\begin{figure}[t]
  \centering
  \includegraphics{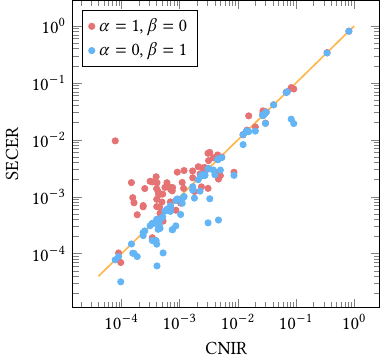}
  \caption{Comparison of cumulative security risk between CNIR and SECER. SECER can substantially reduce security risk relative to CNIR when $\alpha=0,\beta=1$ (blue marks).}
  \Description{Scatter plot comparing the cumulative security risk of paths selected by CNIR on the horizontal axis and SECER on the vertical axis. In security-optimized mode all points lie on or below the diagonal, indicating SECER selects paths with equal or lower risk.}
  \label{fig:exp4:risks}
\end{figure}

\begin{figure}[t]
  \centering
  \includegraphics{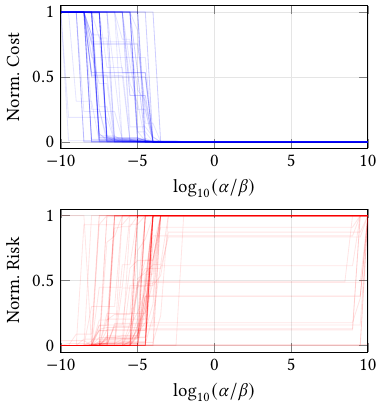}
  \caption{Sensitivity of Normalized Cost and Risk to the log-ratio of weights $\alpha/\beta$ across 100 random node pairs. The overlapping translucent lines indicate the discrete transition of optimal paths in the real Lightning Network topology.}
  \Description{Two stacked line charts showing normalized cost and normalized risk versus the base-10 logarithm of the ratio of alpha to beta, for 100 node pairs drawn as overlapping translucent lines. Most curves transition between the two extremes with discrete jumps concentrated in the interval from minus 8 to minus 3.}
  \label{fig:exp5:sensitivity}
\end{figure}

To evaluate the performance and flexibility of the SECER metric, we conduct a multi-dimensional analysis focusing on economic efficiency, security robustness, and the sensitivity of their trade-off. Using a real-world Lightning Network snapshot, we compare SECER against the CNIR baseline under different configurations of the steering parameters $\alpha$ and $\beta$. For all experiments, we utilize the recursive fee model from Equation~\ref{eq:aumayr_cost} with a target capacity $A = 10,000$ satoshis.

\subsubsection{Economic and Security Extremes} 
We first examine the two extreme configurations of SECER. As shown in Figure~\ref{fig:exp3:costs}, in cost-optimized mode ($\alpha=1, \beta=0$), SECER consistently identifies paths with lower fees than the hop-minimizing CNIR (all points fall below $y=x$), with some reductions exceeding 50\%. Conversely, Figure~\ref{fig:exp4:risks} demonstrates that in security-optimized mode ($\alpha=0, \beta=1$), SECER's risk scores are substantially lower than CNIR. These results confirm that SECER can outperform traditional hop-based metrics in either dimension by effectively navigating the network's continuous risk and fee spectrums.

\subsubsection{Sensitivity Analysis and Pareto Convex Hull}
To analyze the impact of the $\alpha,\beta$ ratio on path selection, we vary $\log_{10}(\alpha/\beta)$ from $-10$ to $10$ for 100 random node pairs in Lightning Network snapshot. Since our solver utilizes linear scalarization for the SECER metric, sweeping these weights effectively samples points on the \textbf{convex hull of the Pareto front} for the cost-risk trade-off. This analysis highlights how the solver's output transitions between different "best" paths as user preferences shift across the spectrum.

As illustrated in Figure~\ref{fig:exp5:sensitivity}, the normalized curves exhibit a clear transition zone, with most path changes occurring within the interval $\log_{10}(\alpha/\beta) \in [-8, -3]$. Within this window, the majority of node pairs undergo a discrete ``jump,'' indicating a direct switch between a security-optimized path and a cost-optimized path, while only a small subset of pairs transitions through multiple intermediate paths. As a single-objective engine performing a discrete parameter sweep, our solver discovers only a partial convex hull and may bypass non-convex Pareto-optimal solutions (see Section \ref{sec:eval:limitations}); this is consistent with our design goal of finding the exact optimal path for any \textit{given} preference in sub-100ms, letting users explore these trade-off dynamics interactively.

\subsubsection{Discussion}
The sharp transition points suggest that SECER does not merely blend objectives linearly, but identifies critical thresholds where a slight shift in preference yields significant gains in the opposite dimension: high-frequency micro-payments can be tuned to the ``low-cost'' side of the transition, while large-value transfers can be shifted toward the ``high-security'' side.

\subsection{Limitations}
\label{sec:eval:limitations}
While our framework demonstrates significant performance and flexibility gains, it is subject to several limitations. First, the exactness guarantee of our quadtree-based solver is contingent on specific metric properties, such as path-monotonicity and capacity-independence. Metrics violating these conditions may require alternative optimization strategies. Second, due to the lack of publicly available real-time balance data for the Lightning Network, our evaluation utilizes public channel capacity as a proxy for available liquidity, following common practices in PCN research. Finally, although SECER provides a tunable preference model, its use of linear scalarization means it primarily explores the convex hull of the multi-objective Pareto front, potentially bypassing non-convex solutions.

\section{Conclusion}
\label{sec:conclusion}

In this paper, we studied the problem of VPC path selection under generalized monotone metrics in Payment Channel Networks. We proposed a structure-aware exact solver based on quadtree search, which exploits monotonicity and distance plateau properties in the capacity-constrained search space to prune large regions while preserving optimality. The resulting approach achieves orders of magnitude speedups over exhaustive methods while maintaining exactness.

To demonstrate the flexibility of the framework, we instantiated it with a composite metric (SECER) that integrates economic cost and security risk, enabling adaptive trade-offs across different application scenarios. Experimental results on real-world Lightning Network snapshots show that our method provides both scalable performance and practical effectiveness for VPC path selection.

Our results suggest several directions for future work. First, the quadtree recursion naturally decomposes into independent subproblems with minimal shared state, making distributed and GPU-accelerated realizations a promising avenue for scaling to even larger networks. Second, extending the solver---exactly or with bounded approximation error---to metrics that violate path-monotonicity or capacity-independence would broaden its applicability to general routing settings. Third, since linear scalarization explores only the convex hull of the Pareto front, integrating $\varepsilon$-constraint or multi-objective search techniques could recover non-convex trade-off solutions. Finally, we plan to develop refined models of economic cost and adversarial risk that capture dynamic channel balances in emerging VPC deployments.

\begin{acks}
  \emergencystretch=1.5em
  This work was supported by the Natural Science Foundation on Frontier Leading Technology Basic Research Project of Jiangsu (No.~BK20222001), the NSFC-62272215 and the 111 Center (No.~B26023). Yuan Zhang and Sheng Zhong are the corresponding authors.\par
\end{acks}

\bibliographystyle{ACM-Reference-Format}
\nocite{*}
\bibliography{bib/references}


\end{document}